\documentclass[runningheads,orivec]{llncs}
\usepackage[T1]{fontenc}
\usepackage{graphicx}
\usepackage[utf8]{inputenc}
\usepackage{booktabs}
\usepackage{algorithm}
\usepackage{algorithmic}
\usepackage{url}
\usepackage{esvect}
\usepackage[hidelinks]{hyperref}
\usepackage{mathtools}

\usepackage{subcaption}

\usepackage{amsmath}
\usepackage{amssymb}
\usepackage{xcolor}

\let\vec\boldsymbol
\AtEndEnvironment{proof}{\qed}
\title{Lexicographic Social Ranking on Monotonic Coalitional Rankings}

\author{Felix Fritz\inst{1}\orcidID{0009-0002-7448-8005} \and
Stefano Moretti\inst{1}\orcidID{0000-0003-3627-3257}}

\authorrunning{F. Fritz \and S. Moretti}
\institute{
LAMSADE, CNRS, Universit\'e Paris-Dauphine, Universit\'e PSL, 75016 Paris, France\\
\email{\{felix.fritz, stefano.moretti\}@dauphine.psl.eu}
}

\newcommand\PS{{2^N}}
\newcommand\card[1]{{\left|#1\right|}}
\newcommand\rank{{\boldsymbol{\mathcal{R}}}}
\newcommand\lex{{\mathrm{lex}}}
\newcommand\LOne{{L^{(1)}}}
\newcommand\MWC{\mathfrak{M}}

\newcommand{\RRlex}[1][\succsim]{\mathrel{R}^{#1}_{\lex}}
\newcommand{\RPlex}[1][\succsim]{\mathrel{P}^{#1}_{\lex}}
\newcommand{\RIlex}[1][\succsim]{\mathrel{I}^{#1}_{\lex}}

\newcommand{\RRlone}[1][\succsim]{\mathrel{R}^{#1}_{\LOne}}
\newcommand{\RPlone}[1][\succsim]{\mathrel{P}^{#1}_{\LOne}}
\newcommand{\RIlone}[1][\succsim]{\mathrel{I}^{#1}_{\LOne}}

\begin{document}

\maketitle

\begin{abstract}
Recent studies on social rankings in coalitional settings have introduced methods that rank individuals by lexicographically comparing vectors of their occurrences across coalitions ordered according to their strength.
In this work, we focus on two such solutions: the lexicographical excellence (lex-cel) solution, which disregards coalition size, and the $\LOne$ solution, which additionally prioritizes smaller coalitions through a double lexicographic comparison.
We investigate the combinatorial connections between these two solutions on monotonic coalitional rankings, where equivalence classes are compactly represented through sets of minimal (with respect to set inclusion) coalitions.
After introducing general formulas for computing the lex-cel and $\LOne$ parameter vectors from these minimal coalitions, we also present worst-case running time results. Finally, to further explore the behavior of the two solutions through simulations designed to assess the distance of the rankings they produce, we show that the monotonicity assumption does not lead to actual redundancy in the rankings produced by the two solutions.
\end{abstract}

\section{Introduction}\label{ch:intro}

The problem of evaluating the importance of individuals based on the performance of the groups they belong to has received increasing attention in the literature on social ranking.
In this framework, a ranking over coalitions of a finite set $N$ of players induces a ranking over the players themselves, reflecting their relative contribution to the coalitions in which they appear.
This problem arises naturally in settings such as voting systems, collaborative research, sports teams, and committee evaluations.
Several social ranking solutions have been proposed and axiomatically characterized, including 
the lexicographic excellence solution (lex-cel)~\cite{bernardi}, and the $\LOne$ solution~\cite{algaba}; see \cite{survey} for a survey on the related literature.

Lex-cel and \(\LOne\) are particularly closely related.
Both evaluate players according to their occurrences in the equivalence classes of a coalitional ranking, proceeding lexicographically from the best coalitions to the worst ones.
While lex-cel only counts the occurrences of players, \(\LOne\) additionally distinguishes coalitions by their cardinalities, giving priority to smaller coalitions within the same equivalence class.

The two solutions also share a substantial axiomatic foundation.
In particular, both satisfy \emph{desirability}, requiring that a player who contributes at least as much as another player to every coalition should not be ranked below them, \emph{independence from the worst set}, stating that refinements of the lowest equivalence class should not reverse an already strict preference relation, and \emph{consistency after indifference}~\cite{desirability}.
At the same time, the two solutions differ in how they treat coalition cardinalities.
While lex-cel is characterized through coalitional anonymity, \(\LOne\) satisfies stronger cardinality-sensitive principles such as \emph{per-size coalitional anonymity} and \emph{\(k\)-desirability on dichotomous rankings}, which explicitly account for the sizes of the coalitions involved.

This strong overlap in their axiomatic structure raises the question of whether the two solutions become more closely aligned on restricted domains of coalitional rankings.
In particular, monotonic rankings impose strong structural constraints on coalitions through inclusion relations, potentially amplifying the common consequences of the axioms shared by both solutions. A coalitional ranking is monotonic if every superset of a coalition is ranked at least as highly as the coalition itself.
Such rankings naturally generalize the notion of {\emph simple games} \cite{owen2013game,alonso2013review} from cooperative game theory and can be represented through the minimal coalitions that reach a given equivalence class.
Monotonicity is a classical property and has been widely studied in the field of coalitional games to understand when classical solutions such as the Shapley value, the Banzhaf value and other semivalues share the same ranking  \cite{carreras2008ordinal,freixas2010ordinal,saari2001some}. More recently, monotonic linear orders over coalitions have been investigated in \cite{GOURVES2026109656} as a condition aimed at simplifying the computation of the ranking provided by the lex-cel (it is easy to show that, on this domain, lex-cel and \(\LOne\) coincide).

A major obstacle in applying social ranking solutions is that the number of coalitions grows exponentially with the number of players.
Recent work has therefore studied ways of reducing the amount of information required to determine social rankings.
Suzuki and Horita~\cite{Suzuki2024}, for instance, investigate social ranking solutions under variable domains of coalitions, motivated by the observation that considering all possible coalitions is often computationally infeasible and practically unrealistic.
Ravier et al.~\cite{ravier2024} study elicitation procedures for identifying lex-cel necessary winners under incomplete knowledge.
In this paper, we take a complementary approach by restricting our attention to monotonic coalitional rankings.
This representation is also cognitively appealing: instead of comparing all coalitions, one may specify coalitions that are just sufficient, in the sense that removing any player makes them worse.
Related models with several levels of approval have been studied in the literature on weighted voting games, i.e., with (2, k) simple games~\cite{freixas2009}.

The role of minimal winning coalitions has been studied extensively in {\it voting power} \cite{owen2013game,alonso2013review}.
Classical power indices such as the Penrose--Banzhaf index \cite{penrose1952,banzhaf1964} and the Shapley--Shubik index \cite{shapley1954} are usually defined using all winning coalitions, but can also be computed from the set of minimal winning coalitions \cite{kirsch2008}.
Similarly, Stach~\cite{Stach2022-ea} reformulates the Public Help index \cite{bertini2008} using null-player-free winning coalitions.

Since monotonic coalitional rankings extend simple games from two classes to several ordered equivalence classes, this suggests asking whether analogous combinatorial reductions are possible for lex-cel and \(\LOne\).
We derive explicit inclusion--exclusion formulas for computing lex-cel and \(\LOne\) scores from the minimal winning coalitions of the upper equivalence classes of a monotonic ranking.
We first treat dichotomous rankings and then extend the formulas to arbitrary monotonic rankings by considering cumulative truncations into simple games.
These formulas also lead to optimization methods: for a fixed player, redundant coalitions can be removed through inclusion-minimal reductions, and for pairwise comparisons it suffices to compute score differences rather than absolute values.
We briefly touch upon their worst-case running times, which marginally improve the exponential growth of the original calculation method.

Finally, we investigate through random sampling whether the additional structure imposed by monotonicity leads lex-cel and $\LOne$ to induce the same ranking more often.
The general conclusion is that, while both solutions align in extreme coalition distributions, neither becomes redundant in general situations.

The remainder of the paper is organized as follows.
Section~\ref{ch:definitions} introduces the relevant notions of coalitional rankings, monotonicity, lex-cel, and \(\LOne\).
Section~\ref{ch:calculations} derives the formulas based on minimal winning coalitions.
Section~\ref{ch:optimization} presents the corresponding optimization methods.
Section~\ref{ch:simulations} reports the simulation results.
Section~\ref{ch:conclusion} concludes.

\section{Basic definitions and concepts}\label{ch:definitions}
In traditional social ranking solutions, elements are ranked based on the ranking of their groups.
Let $N = \{1, \dots, n\}$ be a set of elements and $\PS = \{S \subseteq N\}$ its power set, the set of subsets  or \emph{coalitions} of $N$.
The \emph{cardinality} $\card{S}$ denotes the number of elements in a given coalition $S\subseteq N$.
Here, $\card{N} = n$ and $\card{\PS} = 2^n$.

A \emph{binary relation} on $N$ is a set $R \subseteq N \times N$.
For any $i, j \in N$, an ordered pair $(i,j) \in R$ is usually denoted as $iRj$.
A binary relation $R$ is said to be \emph{transitive} if $iRj$ and $jRk$ implies $iRk$ for all $i,j,k \in N$.
A transitive and total binary relation on a finite set $N$ is called a \emph{total preorder} or a \emph{ranking} on $N$.
Likewise, a \emph{coalitional ranking} describes a total preorder on the power set $2^N$.
 A ranking that is also antisymmetric is called a \emph{linear order}.
The set of all rankings over a set $N$ is denoted by $\rank(N)$.
Consequently, $\succsim\ \in \rank(\PS)$ represents a ranking on all coalitions of elements in $N$.
In such a ranking, $\sim$ indicates its symmetric part ($S \sim T$ if $S \succsim T$ and $T \succsim S$), $\succ$ its asymmetric part ($S \succ T$ if $S \succsim T$ and not $T \succsim S$).
A ranking $\succsim$ may be interpreted as an ordinal representation of the relative strength of coalitions: for example, if the elements represent voters, $S \succ T$ denotes the fact that the coalition $S$ is strictly stronger than  $T$, e.g. because $S$ forms a majority in a house while $T$ does not.

For a given coalitional ranking $\succsim\ \in \rank(\PS)$ represented as $S_1 \succsim S_2 \succsim \dots \succsim S_{2^n}$, its quotient order is written as $\Sigma_1 \succ \Sigma_2 \succ \dots \succ \Sigma_m$.
Each equivalence class $\Sigma_k \subseteq 2^N$ contains coalitions that are all symmetric in $\succsim$.
This means that any coalition in $\Sigma_1$ is indifferent to $S_1$ and strictly better than those in $\Sigma_2, \Sigma_3$, etc.
Likewise, $S_{2^n} \in \Sigma_m$ belongs to the group of lowest ranked coalitions.
Finally, a \emph{social ranking solution}, or \emph{ranking solution}, is a function $R: \rank(\PS) \rightarrow \rank(N)$ yielding a ranking over the set of elements for each ranking over its power set.
A comprehensive survey on ranking solutions can be found in \cite{survey}.

This work extensively studies monotonic power relations.
A coalitional ranking $\succsim\ \in \rank(\PS)$ is \emph{monotonic} under set inclusion if $S \supseteq T$ implies that $S \succsim T$. Notice that a monotonic coalitional ranking $\succsim\ \in \rank(\PS)$ with only two equivalence classes $\Sigma_1 \succ \Sigma_2$
 induces a \emph{simple game} $(N,\Sigma_1)$ \cite{owen2013game}, where coalitions in $\Sigma_1$ are called \emph{winning coalitions} and those in $\Sigma_2$ \emph{losing}. An element of $\Sigma_1$ that is minimal with respect to set inclusion is called a \emph{minimal winning coalition}. These coalitions can be enumerated using algorithms from the literature \cite{MUSalgorithm,zhao2018computing}.

\subsection{The lexicographical excellence and \texorpdfstring{$\LOne$}{L1} solution}
The results in this paper build upon the lexicographical excellence solution \cite{bernardi} and the $\LOne$ solution \cite{algaba}.
Both of these solutions are not concerned with whom players form coalitions, but rather the positions in which these coalitions appear.

The lexicographical excellence solution, short \emph{lex-cel}, takes into account the total number of times a player appears in each equivalence class.
Given two vectors $\vec{i} = (i_1, \dots, i_m)$ and $\vec{j} = (j_1, \dots, j_m)$, we define a lexicographical order $\geq_L$ wherein $\vec{i} \geq_L \vec{j}$ if either $\vec{i} = \vec{j}$ or if there exists a $k$ such that $\vec{i}_k > \vec{j}_k$ and $\vec{i}_\ell = \vec{j}_\ell$ for all $\ell < k$.
In a ranking $\succsim\ \in \rank(\PS)$ with its corresponding quotient order $\Sigma_1 \succ \dots \succ \Sigma_m$, $\theta^{\succsim, i} = (\theta^{\succsim, i}_1, \dots, \theta^{\succsim, i}_m)$ is the $m$-directional vector for a player $i \in N$, also called their lex-cel score, such that each index $k$ describes the number of coalitions in $\Sigma_k$ that $i$ appears in,
\begin{equation}\label{eq:lex}
	\theta^{\succsim, i}_k = \card{\{S \in \Sigma_k: i \in S\}}.
\end{equation}

\begin{definition}\label{def:lexcel}
	The \emph{lexicographical excellence solution} (lex-cel) is the function $\RRlex[]: \rank(\PS) \rightarrow \rank(N)$ defined for any coalitional ranking $\succsim\ \in \rank(\PS)$ as
	\[
		i \RRlex j \quad \text{ if } \quad \theta^{\succsim, i} \geq_L \theta^{\succsim, j}.
	\]
We denote by $\RIlex$ the symmetric part of $\RRlex$ and by $\RPlex$ its asymmetric part.
\end{definition}

Although the lex-cel solution is not concerned with the types of coalitions players can form, the $\LOne$ solution does take into account the sizes of the coalitions.
More concretely, the $\LOne$ solution prioritizes a player's contribution to smaller coalitions over bigger ones.
To formalize this, consider two matrices $\mathcal{A}, \mathcal{B} \in \mathbb{R}^{n \times m}$ and an $\LOne$ order $\geq_\LOne$ such that $\mathcal{A} \geq_\LOne \mathcal{B}$ if either $\mathcal{A} = \mathcal{B}$ or if there exists a row $r$ and column $c$ such that
\begin{enumerate}
	\item $\mathcal{A}_{\hat{r}, \hat{c}} = \mathcal{B}_{\hat{r}, \hat{c}}$ for all $\hat{r} \leq n$ and $\hat{c} < c$,
	\item $\mathcal{A}_{\hat{r}, c} = \mathcal{B}_{\hat{r},c}$ for all $\hat{r} < r$, and
	\item $\mathcal{A}_{r, c} > \mathcal{B}_{r, c}$.
\end{enumerate}

In a ranking $\succsim\ \in \rank(\PS)$ with its corresponding quotient order $\Sigma_1 \succ \dots \succ \Sigma_m$, $M^{\succsim, i}$ is the $n$-by-$m$ matrix for a player $i \in N$, also called their $\LOne$ score, wherein each row $r$ and column $c$ describes the number of $r$-sized coalitions in $\Sigma_c$ that $i$ appears in,
\begin{equation}\label{eq:l1}
	M^{\succsim, i}_{r, c} = \card{\{S \in \Sigma_c: i \in S \text{ and } \card{S} = r\}}.
\end{equation}

\begin{definition}\label{def:LOne}
	The \emph{$\LOne$ solution} is the function $\RRlone[]: \rank(\PS) \rightarrow \rank(N)$ defined for any coalitional ranking $\succsim\ \in \rank(\PS)$ as
	\[
		i \RRlone j \quad \text{ if } \quad M^{\succsim,i} \geq_\LOne M^{\succsim, j}.
	\]
We denote by $\RIlone$ the symmetric part of $\RRlone$ and by $\RPlone$ its asymmetric part.
\end{definition}

\section{Calculations using Minimal Winning Coalitions}\label{ch:calculations}
We will investigate the combinatorial connections between the lexicographical excellence scores and the $\LOne$ scores in monotonic rankings.
In particular, we are interested in cases where only minimal winning coalitions are given, as this could reduce the computational resources required to calculate these scores.

\subsection{Coalitional rankings with two equivalence classes}

In voting theory, the question of determining the power of a given voter using the set of minimal winning coalitions has been studied extensively.
Most notably, the Penrose-Banzhaf index counts the number of coalitions that a voter is critical in \cite{penrose1952,banzhaf1964}.
Note that minimal winning coalitions are coalitions in which \emph{every} voter is critical, but a voter may also be critical in winning coalitions that are not minimal.
Subsequently, Kirsch and Langner have proposed a function to calculate the Penrose-Banzhaf index by only using the set of minimal winning coalitions \cite{kirsch2008}.

In the same vein, we may first consider a monotonic power relation $\succsim$ with two equivalence classes, $\Sigma_1 \succ \Sigma_2$.
Let $\MWC$ be a function that, given an upward-closed set $\Sigma$ (if $S \in \Sigma$ then $T \in \Sigma$ for any $T \supset S$), produces the set of minimal winning coalitions,
\begin{equation}
	\MWC(\Sigma) = \{S \in \Sigma: \nexists T \in \Sigma \text{ s.t. } T \subset S\}.
\end{equation}

Given this kind of input, we begin by showing in Theorem~\ref{thm:l1} how any particular element of a player's $\LOne$ score matrix can be calculated.
For convenience, $[m] = \{1, \dots, m\}$ denotes the set of positive integers up to $m$.

\begin{theorem}\label{thm:l1}
	In a monotonic power relation $\succsim\ \in \rank(\PS)$ with two equivalence classes $\Sigma_1 \succ \Sigma_2$ and its set of minimal winning coalitions $\MWC(\Sigma_1) = \{S_1, \dots, S_m\}$, the $\LOne$ score value of a player $i$ in column $1$ and row $1 \leq r \leq n$ is
	\begin{equation}\label{eq:l1thm}
		M^{\succsim,i}_{r,1} = \sum_{\varnothing \neq L \subseteq [m]} (-1)^{\card{L}-1}\ \binom{n - \card{\{i\} \cup \bigcup_{\ell \in L} S_\ell}}{r - \card{\{i\} \cup \bigcup_{\ell \in L} S_\ell}}.
	\end{equation}
\end{theorem}

\begin{proof}
	Let $\succsim\ \in \rank(\PS)$ be a monotonic coalitional ranking with two equivalence classes, $\Sigma_1 \succ \Sigma_2$, and its set of minimal winning coalitions, $\MWC(\Sigma_1) = \{S_1, \dots, S_m\}$. Recall from Equation~(\ref{eq:l1}) that
	\[
		M^{\succsim,i}_{r,1} = \card{\{S \in \Sigma_1: i \in S \text{ and } \card{S} = r\}}.
	\]

	Take any coalition $S \in \MWC(\Sigma_1)$.
	This coalition has the option of forming a bigger coalition (that is still winning) with any $T \subseteq N \setminus S$ such that $\card{S \cup T} = r$.
	Thus, for any $\card{T} = r - \card{S}$, we may define the set of all coalitions $T'$ subsets of $N \setminus S$ of size $|T|$,
	\begin{equation}\label{eq:l1p1}
		\binom{N \setminus S}{\card{T}} = \big\{ T' \subseteq N \setminus S: \card{T'} = \card{T} \big\}.
	\end{equation}

	Its cardinality is the binomial coefficient,
	\begin{equation}\label{eq:l1p2}
		\card{\binom{N \setminus S}{\card{T}}} = \binom{\card{N \setminus S}}{\card{T}} = \binom{n - \card{S}}{\card{T}}.
	\end{equation}

	Note that the binomial coefficient is well defined when choosing a ``negative'' number of elements from a set, namely $\binom{n}{k} = 0$ for any $k < 0$ (it also yields $0$ for $k > n$).

	Next, we take into account that player $i$ must be part of these coalitions.
	Inspired by principal filters, this set will be notated as
	\begin{equation}
		\uparrow^{+i}_r S = \{ V \subseteq N : V \supseteq S \text{ and } i \in V \text{ and } \card{V} = r\}.
	\end{equation}

	To determine the cardinality $\card{\uparrow^{+i}_r S}$, we differentiate between two cases:
	\begin{enumerate}
		\item If $i \in S$, thus $i \in S \cup T$ for any $T \in \binom{N \setminus S}{\card{T}}$, then Equations~(\ref{eq:l1p1}) and (\ref{eq:l1p2}) hold;
		\item If $i \notin S$, then $i \in N \setminus S$ and $T$ must include $i$.
		      Therefore, Equation~(\ref{eq:l1p1}) changes to
		      \begin{equation}
			      \binom{N \setminus (\{i\} \cup S)}{\card{T} - 1}
		      \end{equation}
		      and its corresponding cardinality becomes
		      \begin{equation}
			      \card{\binom{N \setminus (\{i\} \cup S)}{\card{T} - 1}} = \binom{n - \card{\{i\} \cup S}}{\card{T} - 1}.
		      \end{equation}

	\end{enumerate}
	Substituting $\card{T}$ with $r - \card{S}$, we can simplify
	\begin{align*}
		\card{\uparrow^{+i}_r S}                   & = \begin{cases}
			                                               \binom{n - \card{S}}{r - \card{S}}                & , i \in S,    \\[0.3cm]
			                                               \binom{n - \card{\{i\} \cup S}}{r - \card{S} - 1} & , i \notin S;
		                                               \end{cases} \\
                                                       \Rightarrow \card{\uparrow^{+i}_r S} & = \begin{cases}
			                                               \binom{n - \card{\{i\} \cup S}}{r - \card{\{i\} \cup S}} & , i \in S,    \\[0.3cm]
			                                               \binom{n - \card{\{i\} \cup S}}{r - \card{\{i\} \cup S}} & , i \notin S;
		                                               \end{cases} \\[0.3cm]
		\Rightarrow \card{\uparrow^{+i}_r S} & = \binom{n - \card{\{i\} \cup S}}{r - \card{\{i\} \cup S}}.
	\end{align*}

	Observe that the union over all $r$-sized principal filters containing $i$, $\bigcup^m_{k=1} \uparrow^{+i}_r S_k$, is precisely the set of coalitions considered by $\LOne$, namely
	\[
		M^{\succsim,i}_{r,1} = \card{\bigcup^m_{k=1} \uparrow^{+i}_r S_k}.
	\]
	Since we have established how to compute the cardinality of $\uparrow^{+i}_r S_k$ for a single $S_k$, we can apply the inclusion--exclusion principle to calculate the cardinality of the union.
	The principle states that, for any finite sets $A_1, \dots, A_n,$
	\[
		\card{\bigcup_{k=1}^n A_k} = \sum_{\varnothing \neq L \subseteq [n]} (-1)^{\card{L}-1}\ \card{\bigcap_{\ell \in L} A_\ell}.
	\]
	To conclude our proof, we only need to determine the set
	\begin{align*}
		\bigcap_{k=1}^m \uparrow^{+i}_r S_k & = \bigcap_{k=1}^m \Big\{T \in \PS: T \supseteq S_k \text{ and } i \in T \text{ and } \card{T} = r\Big\} \\
		                                    & = \Big\{T \in \PS: T \supseteq \bigcup_{k=1}^m S_k \text{ and } i \in T \text{ and } \card{T} = r\Big\} =\ \uparrow^{+i}_r \bigcup_{k=1}^m S_k,
	\end{align*}
	and its cardinality
    \[
    \card{\uparrow^{+i}_r \bigcup_{k=1}^m S_k} = \binom{n - \card{\{i\} \cup \bigcup_{k = 1}^m S_k}}{r - \card{\{i\} \cup \bigcup_{k = 1}^m S_k}}.
    \]
    
    Although this considers specifically the union over the principal filters of all minimal winning coalitions, this same procedure applies to any of its subsets, that is, for some $\varnothing \neq L \subseteq [m]$,
    \[
        \bigcap_{\ell \in L} \uparrow^{+i}_r S_\ell = \uparrow^{+i}_r \bigcup_{\ell \in L} S_\ell.
    \]
	Applying the inclusion--exclusion principle, we obtain the desired expression in Equation~(\ref{eq:l1thm}), concluding our proof of Theorem~\ref{thm:l1}:
	\begin{align*}
		M_{r,1}^{\succsim, i} & = \card{\bigcup_{k=1}^m \uparrow^{+i}_r S_k} \\
        &= \sum_{\varnothing \neq L \subseteq [m]} (-1)^{\card{L} - 1} \ \card{\bigcap_{\ell \in L} \uparrow_r^{+i} S_\ell}                                                               \\
		                      & = \sum_{\varnothing \neq L \subseteq [m]} (-1)^{\card{L} - 1} \ \card{\uparrow_r^{+i} \bigcup_{\ell \in L} S_\ell} \\
                              &= \sum_{\varnothing \neq L \subseteq [m]} (-1)^{\card{L}-1}\ \binom{n - \card{\{i\} \cup \bigcup_{\ell \in L} S_\ell}}{r - \card{\{i\} \cup \bigcup_{\ell \in L} S_\ell}}.
	\end{align*}
\end{proof}

\begin{theorem}\label{thm:lex}
	In a monotonic power relation $\succsim\ \in \rank(\PS)$ with two equivalence classes $\Sigma_1 \succ \Sigma_2$ and its set of minimal winning coalitions $\MWC(\Sigma_1) = \{S_1, \dots, S_m\}$, the first lex-cel value of a player $i$ is
	\begin{equation}\label{eq:lexthm}
		\theta^{\succsim,i}_1 = \sum_{\varnothing \neq L \subseteq [m]} (-1)^{\card{L}-1}\ 2^{n-\card{\{i\}\ \cup\ \bigcup_{\ell \in L} S_\ell}}.
	\end{equation}
\end{theorem}

\begin{proof}
	This proof builds upon Theorem~\ref{thm:l1}.

	First, observe the following relationship between the sets defined in Equation~(\ref{eq:lex}) and Equation~(\ref{eq:l1}):
	\begin{equation}\label{eq:lexl1}
		\{S \in \Sigma_1: i \in S\} = \bigcup_{r=1}^n \{S \in \Sigma_1: i \in S \text{ and } \card{S} = r\}.
	\end{equation}

	Since each individual set on the right hand side of Equation~(\ref{eq:lexl1}) is disjoint, it follows that
	\[
		\card{\{S \in \Sigma_1: i \in S\}} = \sum_{r=1}^n \card{\{S \in \Sigma_1: i \in S \text{ and } \card{S} = r\}}.
	\]

	Second, recall the well-known binomial identity $\sum_{k=0}^n \binom{n}{k} = 2^n$.

	Note that, as mentioned in the proof of Theorem~\ref{thm:l1}, a negative lower index in a binomial coefficient yields zero.
	Therefore, the extended sum $\sum_{k=-\infty}^n \binom{n}{k}$ also produces $2^n$.

	For any constant $c$ such that $1 \leq c \leq n$, we can determine that
	\[
		\sum_{k=1}^n \binom{n-c}{k-c} = 2^{n-c}.
	\]

	Thus, we conclude:
	\begin{align*}
		\theta^{\succsim,i}_1 & = \sum_{r=1}^n M^{\succsim,i}_{r,1}                                                                                                                                                    \\[.2cm]
		                      & = \sum_{r=1}^n \sum_{\varnothing \neq L \subseteq [m]} (-1)^{\card{L}-1}\ \binom{n - \card{\{i\} \cup \bigcup_{\ell \in L} S_\ell}}{r - \card{\{i\} \cup \bigcup_{\ell \in L} S_\ell}} \\[.2cm]
		                      & = \sum_{\varnothing \neq L \subseteq [m]} (-1)^{\card{L}-1}\ \sum_{r=1}^n \binom{n - \card{\{i\} \cup \bigcup_{\ell \in L} S_\ell}}{r - \card{\{i\} \cup \bigcup_{\ell \in L} S_\ell}} \\[.2cm]
		                      & = \sum_{\varnothing \neq L \subseteq [m]} (-1)^{\card{L}-1}\ 2^{n-\card{\{i\}\ \cup\ \bigcup_{\ell \in L} S_\ell}}.
	\end{align*}
\end{proof}

\subsection{Calculations on any monotonic coalitional ranking}

Given a number of players $n$, the total number of coalitions containing an element $i$ remains constant.
Therefore, it is easy to verify that, in dichotomous rankings $\succsim\ \in \rank(2^N)$ with $\Sigma_1 \succ \Sigma_2$,
\begin{equation*}
	\theta^{\succsim, i}_2 = 2^{n-1} - \theta^{\succsim,i}_1\quad\text{ and }\quad M^{\succsim,i}_{r,2} = \binom{n-1}{r-1} - M^{\succsim,i}_{r,1}.
\end{equation*}

Extending the set of monotonic power relations in $\rank(2^N)$ to partition into any number of equivalence classes, consider the set of minimal winning coalitions in $\Sigma_k$ to be
\begin{equation*}
	\MWC(\Sigma_1 \cup \dots \cup \Sigma_k) = \{S_1, \dots, S_m\}.
\end{equation*}

Obviously,
\begin{equation*}
	\Sigma_k = (\Sigma_1 \cup \dots \cup \Sigma_k) \setminus (\Sigma_1 \cup \dots \cup \Sigma_{k-1}).
\end{equation*}

Therefore,
\begin{align*}
	\theta^{\succsim, i}_k &= \sum_{\varnothing \neq L \subseteq [m]} (-1)^{\card{L}-1}\ 2^{n-\card{\{i\}\ \cup\ \bigcup_{\ell \in L} S_\ell}} - (\theta^{\succsim,i}_1 + \dots + \theta^{\succsim,i}_{k-1})\text{, and}\\
	M^{\succsim,i}_{r,k} &= \sum_{\varnothing \neq L \subseteq [m]} (-1)^{\card{L}-1}\ \binom{n - \card{\{i\} \cup \bigcup_{\ell \in L} S_\ell}}{r - \card{\{i\} \cup \bigcup_{\ell \in L} S_\ell}} - (M^{\succsim,i}_{r,1} + \dots + M^{\succsim,i}_{r,k-1}).
\end{align*}

In other words, for each index $k$, a monotonic power relation can be reduced to a simple game by treating all coalitions in $\Sigma_1 \cup \dots \cup \Sigma_k$ as winning and all remaining coalitions as losing.
Since at each step we overcount player $i$ by the number of appearances in $\Sigma_1$ through $\Sigma_{k-1}$, subtracting those suffices to produce the exact number of times $i$ appears in $\Sigma_k$.

\begin{example}\label{ex:calc}
	Consider a coalitional ranking $\succsim\ \in \rank(2^N)$ such that\footnote{In the following, we omit braces and commas to specify a set of elements: for instance, for coalition $\{1,2,3\}$ we simply write $123$.}
	\[
		\begin{alignedat}{1}
			 & (12345 \sim 1234 \sim 1235 \sim 1245 \sim 1345 \sim 2345                                                                 \\
			 & \quad \sim 123 \sim 134 \sim 135 \sim \boldsymbol{234} \sim \boldsymbol{235} \sim \boldsymbol{245} \sim \boldsymbol{13})
			\succ \boldsymbol{23}
			\succ \Sigma_3
		\end{alignedat}
	\]
	with $\Sigma_3$ containing all remaining coalitions.
	The corresponding sets of minimal winning coalitions (shown in bold in the ranking) are
	\begin{align*}
		\mathfrak{M}(\Sigma_1)                             & = \{13, 234, 235, 245\}, \\
		\mathfrak{M}(\Sigma_1 \cup \Sigma_2)               & = \{13, 23, 245\},       \\
		\mathfrak{M}(\Sigma_1 \cup \Sigma_2 \cup \Sigma_3) & = \{\varnothing\}.
	\end{align*}
	Because $\{1, 3\} \in \Sigma_1$ is the only coalition of size $2$, the $\LOne$ ranking states that
	\[
		1 \RPlone 2.
	\]
	Following Theorem \ref{thm:lex}, the lex-cel value $\theta^{\succsim,1}_1$ can be calculated as follows:
	\begin{alignat*}{2}
		\theta^{\succsim,1}_1
		 & = 2^{n-\card{\{1, 3\}}} + 2^{n-\card{\{1,2, 3, 4\}}} + 2^{n-\card{\{1,2, 3, 5\}}} + 2^{n-\card{\{1,2, 4, 5\}}}
		 &                                                                                                                & \quad (\card{L}=1) \\
		 & \quad - (2^{n-\card{\{1,2,3,4\}}} + 2^{n-\card{\{1,2,3,5\}}} + 4 * 2^{n-\card{\{1,2,3,4,5\}}})
		 &                                                                                                                & \quad (\card{L}=2) \\
		 & \quad + 4 * 2^{n-\card{\{1,2,3,4,5\}}} - 2^{n-\card{\{1,2,3,4,5\}}}
		 &                                                                                                                & \quad (\card{L} = 3, \card{L}=4) \\
		 & = (2^3 + 3*2^1) - (2*2^1 + 4*2^0) + 4*2^0 - 2^0                                                                                     = 9.
	\end{alignat*}
	Coincidentally, $\theta_1^{\succsim,2}$ is also $9$.
	Advancing to the second equivalence class, since it only contains $\{2,3\}$, lex-cel then determines the opposite preference to $\LOne$,
	\[
		2 \RPlex 1.
	\]
\end{example}

\section{Optimization methods}\label{ch:optimization}

While these calculation methods can be convenient, they become infeasible when the number of minimal winning coalitions is large. Since this set forms an antichain, Sperner's theorem \cite{sperner} implies that its maximum cardinality is
\[
	\max \card{\MWC(\Sigma)} = \binom{n}{\lfloor n/2 \rfloor},
\]
with $\Sigma = \Sigma_1 \cup \dots \cup \Sigma_k$ for some $k \geq 1$.
Therefore, since the formulas above require summation over all nonempty subsets of \(\MWC(\Sigma)\), the worst-case running time is
\[
	O\!\left(2^{\binom{n}{\lfloor n/2 \rfloor}}\right).
\]

We would like to highlight two methods that can reduce the number of coalitions that need to be considered.

First, observe that in Equations (\ref{eq:l1thm}) and (\ref{eq:lexthm}), only unions of the form
\[
\{i\} \cup \bigcup_{\ell \in L} S_\ell
\]
are relevant.

Now suppose that \(S \in \MWC(\Sigma)\) satisfies \(i \in S\), and that for some \(T \in \MWC(\Sigma)\) with \(i \notin T\) we have $S \subseteq \{i\} \cup T$.
Then every superset of \(\{i\} \cup T\) is also a superset of \(S\), so the contribution of \(T\) is already covered by \(S\). Removing \(T\) therefore does not affect the calculation.

More generally, adding $i$ to any set in $\MWC(\Sigma)$ does not alter the calculations.
Therefore, it suffices to consider only the inclusion-minimal sets in
\[\{\{i\} \cup S: S \in \MWC(\Sigma)\}.\]
That is, the relevant family is \[\mathcal F^i(\Sigma)
	= \MWC(\{\{i\} \cup S: S \in \MWC(\Sigma)\}).\]

\begin{example}
	Continuing Example \ref{ex:calc}, recall that $\MWC(\Sigma_1) = \{13, 234, 235, 245\}$ and $\MWC(\Sigma_1 \cup \Sigma_2) = \{13, 23, 245\}$.
	Then, to calculate the lex-cel and $\LOne$ values for the elements $1$ and $2$, it suffices to consider the following sets:
	\begin{align*}
		\mathcal{F}^1(\Sigma_1) & = \{13, 1245\},            & \mathcal{F}^1(\Sigma_1 \cup \Sigma_2) & = \{13, 1245\}, \\
		\mathcal{F}^2(\Sigma_1) & = \{123, 234, 235, 245 \}, & \mathcal{F}^2(\Sigma_1 \cup \Sigma_2) & = \{23, 245\}.
	\end{align*}
\end{example}

\begin{proposition}\label{prp:opt1}
	For a given $n$, the worst-case running time to calculate a lex-cel or $\LOne$ score for an element $i$ given its inclusion-minimal set $\mathcal{F}^i(\Sigma)$ is
	\[
		O\!\left(2^{\binom{n-1}{\lfloor n/2 \rfloor - 1}}\right).
	\]
\end{proposition}

The proof of Proposition \ref{prp:opt1} can be found in the Appendix.

A second improvement to consider is that social ranking solutions are generally concerned with ordering only pairs of elements, not calculating some index value.
Thus, if a comparison between $i$ and $j$ depends on some absolute values $v_i$ and $v_j$, the difference $v_i - v_j$ suffices to compare these two.

Without loss of generality, let $\succsim\ \in \rank(2^N)$ with $\Sigma_1 \succ \Sigma_2$ be a dichotomous ranking, and $\MWC(\Sigma_1) = \{S_1, \dots, S_m\}$.
Then, lex-cel determines that
\[
	\begin{cases}
		i \RPlex j & \text{if } \theta^{\succsim,i}_1 - \theta^{\succsim,j}_1 > 0, \\
		j \RPlex i & \text{if } \theta^{\succsim,i}_1 - \theta^{\succsim,j}_1 < 0, \\
		i \RIlex j & \text{otherwise.}
	\end{cases}
\]

For brevity, denote $S_L = \bigcup_{\ell \in L} S_\ell$. We calculate
\begin{align*}
	 & \theta^{\succsim,i}_1 - \theta^{\succsim,j}_1                                                                                                                                                                                         \\
	 & = \sum_{\varnothing \neq L \subseteq [m]} (-1)^{\card{L}-1}\ 2^{n-\card{\{i\}\ \cup\ S_L}} - \sum_{\varnothing \neq L \subseteq [m]} (-1)^{\card{L}-1}\ 2^{n-\card{\{j\}\ \cup\ S_L}} \\
	 & = \sum_{\varnothing \neq L \subseteq [m]} (-1)^{\card{L}-1}\ \Big(2^{n-\card{\{i\}\ \cup\ S_L}} - 2^{n-\card{\{j\}\ \cup\ S_L}}\Big).
\end{align*}

The subtraction indicates that the relation between $i$ and $j$ depends only on the unions over $S_\ell$ that contain either $i$ or $j$ but not both,
\[
	2^{n-\card{\{i\}\cup S_L}} -
		2^{n-\card{\{j\}\cup S_L}} =
	\begin{cases}
		\hphantom{-}2^{n - 1 - \card{S_L}} & \text{if } i \in S_L \not\ni j, \\
		-2^{n - 1 - \card{S_L}}            & \text{if } i \notin S_L \ni j,  \\
		\hphantom{-}0                                              & \text{otherwise.}
	\end{cases}
\]

\begin{proposition}
Given two elements $i$ and $j$, the difference between their lex-cel score value in a given monotonic ranking $\succsim\ \in \rank(2^N)$ is
\begin{align}
	 & \theta^{\succsim,i}_1 - \theta^{\succsim,j}_1                                                                                                                                              \label{eq:lexdiff} \\
	 & = \sum_{\substack{\varnothing \neq L \subseteq [m] \\ \mathclap{\{i,j\} \cap S_L = \{i\}}}} (-1)^{\card{L}-1}\ 2^{n - 1 - \card{S_L}}
	 \quad - \sum_{\substack{\varnothing \neq L \subseteq [m]\\ \mathclap{\{i,j\} \cap S_L = \{j\}}}} (-1)^{\card{L}-1}\ 2^{n - 1 - \card{S_L}}. \nonumber
\end{align}
\end{proposition}

For the $\LOne$ solution, we can apply Pascal's rule for binomial coefficients, which states that
\[
	\binom{n}{k} - \binom{n-1}{k-1} = \binom{n-1}{k}.
\]

As before, we calculate the difference between two $\LOne$ score values by determining
\[
	\binom{n-\card{\{i\}\cup S_L}}{r-\card{\{i\}\cup S_L}} -
	\binom{n-\card{\{j\}\cup S_L}}{r-\card{\{j\}\cup S_L}} =
	\begin{cases}
		\hphantom{-} \binom{n-\card{S_L}-1}{r-\card{S_L}} & \text{if } i \in S_L \not\ni j, \\
		- \binom{n-\card{S_L}-1}{r-\card{S_L}}            & \text{if } i \notin S_L \ni j, \\
		\hphantom{-} 0                                    & \text{otherwise.}
	\end{cases}
\]

\begin{proposition}
Given two elements $i$ and $j$, the difference between their $\LOne$ score value in a given monotonic ranking $\succsim\ \in \rank(2^N)$ is
\begin{align}
	 & M^{\succsim,i}_{r,c} - M^{\succsim,j}_{r,c} \\
	 & =
	\sum_{\substack{\varnothing \neq L \subseteq [m] \\ \mathclap{\{i,j\} \cap S_L = \{i\}}}} (-1)^{\card{L}-1}\ \binom{n - \card{S_L} - 1}{r - \card{S_L}} -
	\sum_{\substack{\varnothing \neq L \subseteq [m] \\ \mathclap{\{i,j\} \cap S_L = \{j\}}}} (-1)^{\card{L}-1}\ \binom{n - \card{S_L} - 1}{r - \card{S_L}}. \nonumber
\end{align}
\end{proposition}

\begin{example}
	Continuing Example \ref{ex:calc}, let $S_1 = \{1,3\}$ and $S_2$, $S_3$, and $S_4$ be assigned to the remaining elements of $\MWC(\Sigma_1)$.
	The only subsets $L \subseteq [m]$ to consider when comparing $1$ against $2$ are
	\begin{align*}
		\theta^{\succsim,1}_1 - \theta^{\succsim,2}_1 & = \sum_{\mathclap{L \in \{\{1\}\}}} (-1)^{\card{L}-1}\ 2^{n-1-\card{S_L}} - \sum_{\mathclap{\varnothing \neq L \subseteq \{2,3,4\}}} (-1)^{\card{L}-1}\ 2^{n-1-\card{S_L}} \\
        & = 2^2 - (3*2^1-3*2^0+2^0) = 0.
	\end{align*}
\end{example}

\begin{proposition}\label{prp:opt2}
	For a given $n$, the worst-case running time to calculate the difference of a lex-cel or $\LOne$ score between two elements is
	\[
		O\!\left( 2 \cdot 2^{\binom{n-2}{\lfloor n/2 \rfloor - 1}} \right).
	\]
\end{proposition}

The proof of Proposition \ref{prp:opt2} can be found in the Appendix.

\section{Simulations}\label{ch:simulations}

Monotonicity puts a significant constraint on the way coalitional rankings can be constructed.
Considering simple games, for example, the Dedekind number tells us that there are only $168$ different possible configurations for $n=4$, while there would be $2^{2^n} = 65{,}536$ ways to pick between $16$ coalitions in total.
Allowing for any number of equivalence classes, the gap widens more drastically: there are about
$1.7*10^{9}$ monotonic coalitional rankings, in contrast to $5.3*10^{15}$ without monotonicity (the $n$-th ordered Bell number).

While the lex-cel and $\LOne$ solutions are defined for arbitrary coalitional rankings, restricting the domain to monotonic rankings raises the question of whether these structural constraints lead to a systematic alignment of the two solutions.

In particular, we investigate whether monotonicity reduces the distinction between $\RRlex[]$ and $\RRlone[]$, potentially resulting in identical rankings over $N$. If so, the additional structural information used by the $\LOne$ solution would become redundant in monotonic settings.

To this end, we conduct simulations based on randomly generated monotonic coalitional rankings with up to three equivalence classes.
We record two types of discrepancies: \emph{tie breaks}, where lex-cel yields $i \RIlex j$ while $\LOne$ yields $i \RPlone j$, and \emph{inversions}, where $i \RPlex j$ but $j \RPlone i$.

Furthermore, we analyze how the outcomes depend on the sizes of the top equivalence classes $\Sigma_1$ and $\Sigma_2$, with $\Sigma_3 = 2^N \setminus (\Sigma_1 \cup \Sigma_2)$ containing the remaining coalitions.

\begin{figure}[t]
	\centering
	\begin{subfigure}[b]{0.9\textwidth}
		\centering
		\includegraphics[width=\textwidth]{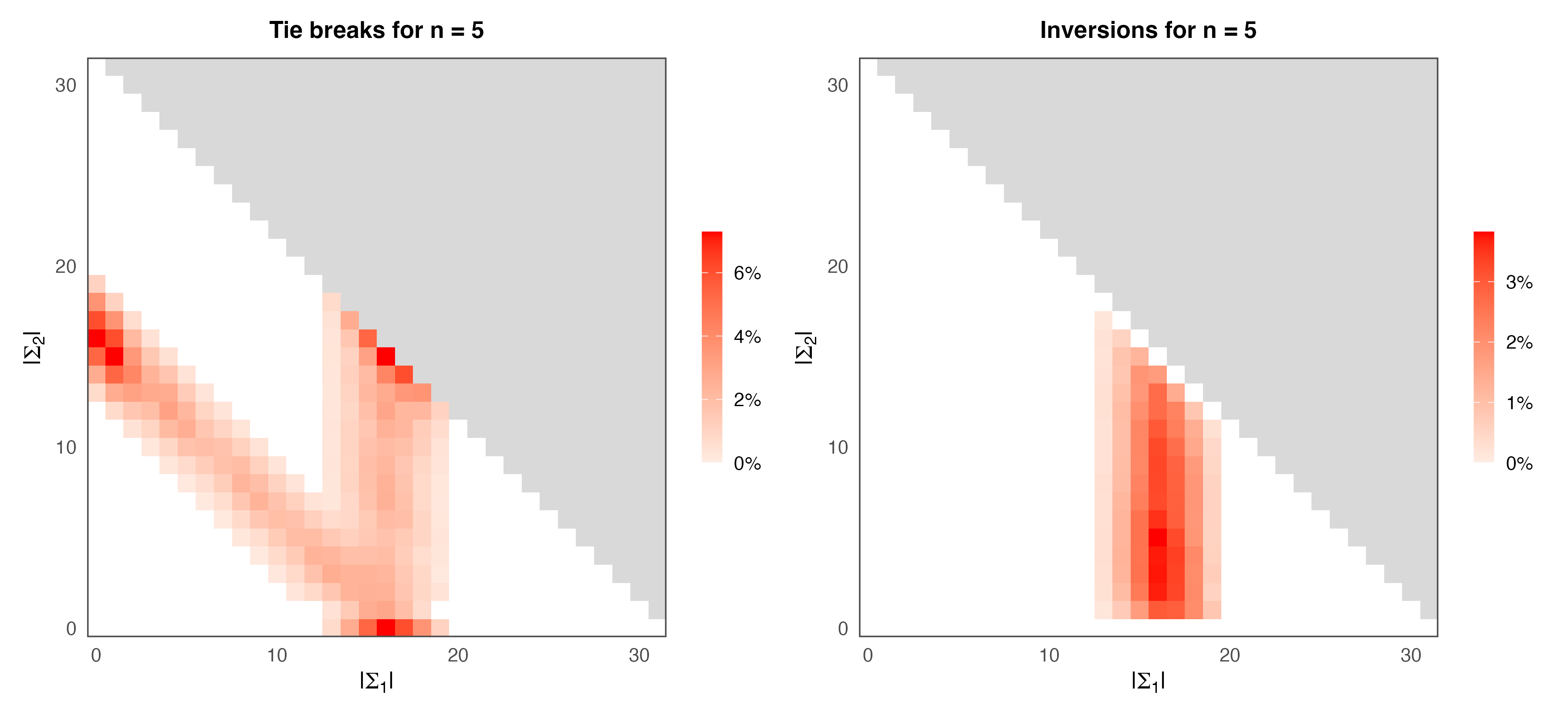}
	\end{subfigure}

	\begin{subfigure}[b]{0.9\textwidth}
		\centering
		\includegraphics[width=\textwidth]{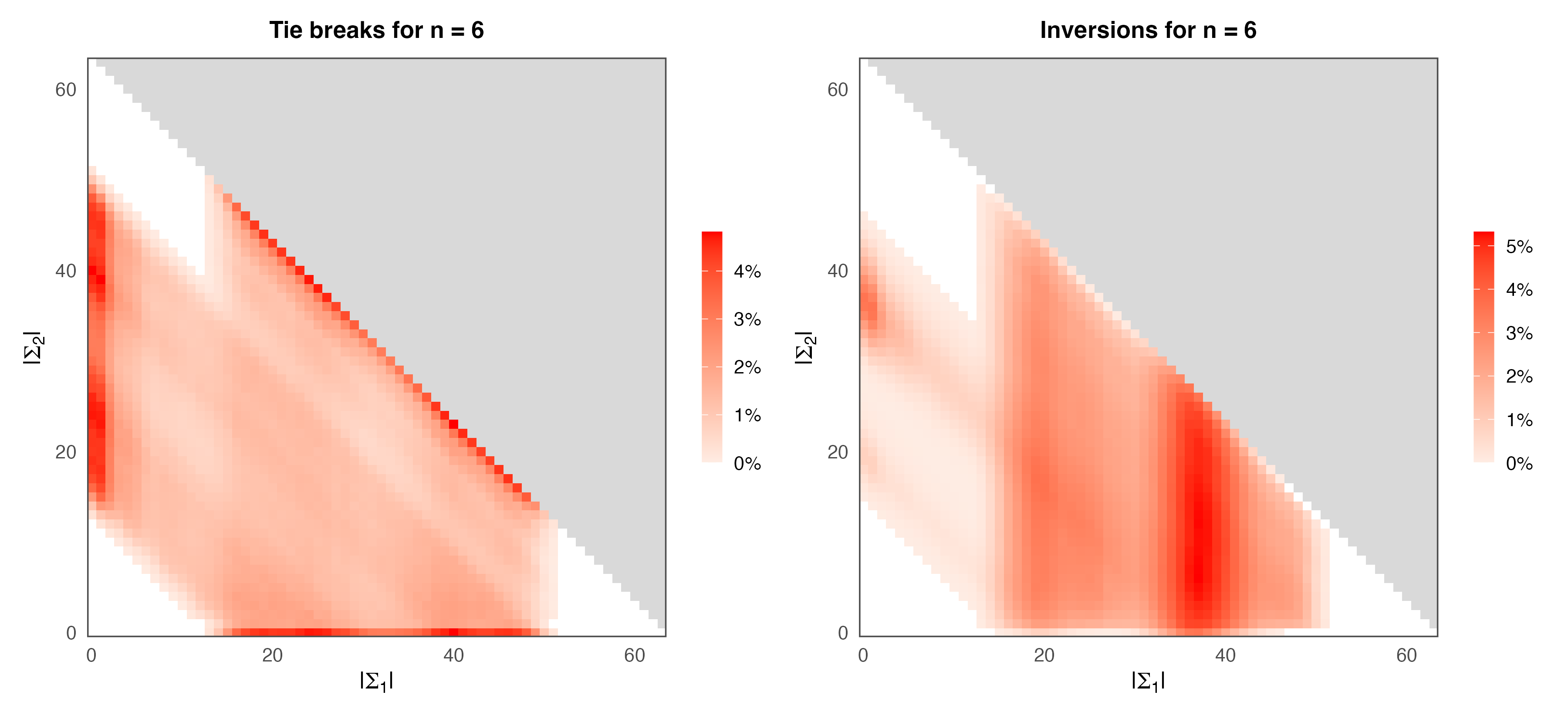}
	\end{subfigure}

	\caption{Simulations comparing the rate of tie breaks and inversions between the lex-cel and $\LOne$ ranking solutions for $n=5$ and $n=6$. Each data point represents 5000 randomly generated rankings.}
	\label{fig:simulation_grid}
\end{figure}

Figure~\ref{fig:simulation_grid} summarizes the results for $n=5$ and $n=6$.
Each point represents the proportion of pairs $(i,j)$ over 5000 coalitional rankings for which both solutions agree, i.e., either $i I j$ or $i P j$, versus cases where they differ (through tie breaks or inversions).

Several patterns emerge from the simulations. First, we observe white triangular regions at the corners of each plot, which persist across all tested values of $n$. These regions correspond to extreme distributions of coalitions among the equivalence classes. In particular, identical rankings are obtained whenever one of the following holds:

\begin{enumerate}
    \item $\card{\Sigma_1} + \card{\Sigma_2} \leq 12$,
    \item $\card{\Sigma_1} + \card{\Sigma_2} \geq 2^n - 12$,
    \item $\card{\Sigma_1} \geq 2^n - 12$.
\end{enumerate}

Notice that for $n \leq 4$, $2^n - 12 < 12$.
Through exhaustive search, we can confirm that lex-cel and $\LOne$ coincide in these smaller cases.

Second, for $n=5$, inversions do not occur at the boundaries of the plots, i.e., whenever one equivalence class is empty. Thus, inversions only arise when all three equivalence classes are non-empty, whereas tie breaks may still occur at the boundaries.
Third, tie breaks are concentrated near the boundaries of the parameter space. These correspond to cases where either $\Sigma_1$ or $\Sigma_2$ is very small, or where $\Sigma_1 \cup \Sigma_2$ contains almost all coalitions. This behavior is consistent with the independence of the worst set property: if one solution yields a decisive comparison based on higher-ranked coalitions, the other solution may reverse this comparison based on lower-ranked ones (which, in turn, do not affect the decision of the former ranking solution).

Finally, although the overall frequency of discrepancies remains below $10\%$, it increases with $n$. This suggests that the $\LOne$ solution does not become redundant in monotonic settings, as differences between the two solutions persist and become more pronounced for larger systems.

\section{Conclusion}\label{ch:conclusion}

Monotonicity is a natural property in many situations where coalitions are formed, whether in political environments, collaborative work settings, or more general object-based frameworks.
The goal of this paper was to investigate alternative methods for applying social ranking solutions within this setting.
We presented combinatorial formulas to calculate the lex-cel and $\LOne$ scores in monotonic coalitional rankings using only the corresponding sets of minimal winning coalitions.
Additionally, motivated by the substantial overlap in the axiomatic foundations of the two solutions, we investigated whether restricting the domain to monotonic rankings leads to a stronger alignment of their induced rankings.

The methods developed in this paper may also be applicable to other social ranking solutions.
At the same time, we acknowledge that the presented formulas still exhibit exponential worst-case running times when the number of minimal winning coalitions becomes large.
Thus, the main contribution of this work is not necessarily to provide computationally optimal algorithms, but rather to establish a combinatorial framework through which social ranking solutions can be studied on monotonic domains.

If this framework based on the inclusion--exclusion principle is applicable to other social ranking solutions, future work could investigate unified computational approaches capable of evaluating several social ranking solutions simultaneously by reusing these shared combinatorial computations.
Moreover, since social rankings are primarily concerned with ordinal comparisons between elements rather than exact score values, monotonic coalitional rankings could enable more efficient methods of ranking elements, possibly by introducing entirely new social ranking solutions.

As suggested by a reviewer,  another promising direction for future research is to investigate a framework focusing on coalitional rankings that are monotonic but incomplete \cite{ravier2024,Suzuki2024}, in particular regarding procedural complexity and potential algorithmic simplifications.
Finally, an axiomatic study specifically tailored to monotonic coalitional rankings could further clarify the behavioral relationships among existing social ranking solutions.


\subsubsection{\ackname} 
S. Moretti acknowledges financial support from the ANR project THEMIS (ANR-20-CE23-0018) and GATSBII (ANR-24-CE23-6645).



\bibliographystyle{splncs04}
\bibliography{refs}

\newpage

\appendix

\section{Appendix}

\subsection{Proof of Proposition~\ref{prp:opt1}}

\begin{proof}
	It was shown previously that the maximum cardinality of an antichain is $\binom{n}{\lfloor n/2 \rfloor}$.
	We construct this antichain by choosing for all coalitions of size $\lfloor n/2 \rfloor$,
	\begin{align*}
		\Sigma       & = \{S \subseteq N: \card{S} \geq \lfloor n/2 \rfloor\}, \\
		\MWC(\Sigma) & = \{S \in \Sigma\,: \card{S} = \lfloor n/2 \rfloor\}.
	\end{align*}
	Partition the set $\MWC(\Sigma)$ into two sets,
	\[
		\MWC^{+i}(\Sigma) = \{S \in \MWC(\Sigma): i \in S\}, \quad \MWC^{-i}(\Sigma) = \{S \in \MWC(\Sigma): i \notin S\}.
	\]
	Put differently, given $N \setminus \{i\}$, $\MWC^{+i}(\Sigma)$ contains all coalitions of size $\lfloor(n/2)-1\rfloor$ for $i$ to form a coalition with, while in $\MWC^{-i}(\Sigma)$ coalitions of size $\lfloor(n/2)\rfloor$ without $i$ are present.
	Their cardinality can be expressed again with the binomial coefficient,
	\[
		\card{\MWC^{+i}(\Sigma)} = \binom{n-1}{\lfloor n/2 \rfloor - 1} \leq \binom{n-1}{\lfloor n/2 \rfloor} = \card{\MWC^{-i}(\Sigma)}.
	\]
	Note that $\card{\MWC^{+i}(\Sigma)} = \card{\MWC^{-i}(\Sigma)}$ if $n$ is even.

	By definition, any $\{i\} \cup S$ for $S \in \MWC^{-i}(\Sigma)$ must be a superset of some coalition in $\MWC^{+i}(\Sigma)$.
	Therefore,
	\[
		\mathcal{F}^i(\Sigma) = \MWC^{+i}(\Sigma),
	\]
	which effectively reduces the original exponential worst-case running time to the maximum cardinality that $\mathcal{F}^i(\Sigma)$ can have,
	\[
		O\!\left(2^{\binom{n-1}{\lfloor n/2 \rfloor - 1}}\right).
	\]
\end{proof}

\subsection{Proof of Proposition~\ref{prp:opt2}}

\begin{proof}
	Fix two distinct elements $i,j \in N$ and let $\MWC(\Sigma)=\{S_1,\dots,S_m\}$.
	Partition the minimal winning coalitions according to the presence of $i$ and $j$:
	\begin{align*}
		\mathfrak M^{+i,+j}(\Sigma) & = \{S \in \MWC(\Sigma): i \in S,\ j \in S\},       \\
		\mathfrak M^{+i,-j}(\Sigma) & = \{S \in \MWC(\Sigma): i \in S,\ j \notin S\},    \\
		\mathfrak M^{-i,+j}(\Sigma) & = \{S \in \MWC(\Sigma): i \notin S,\ j \in S\},    \\
		\mathfrak M^{-i,-j}(\Sigma) & = \{S \in \MWC(\Sigma): i \notin S,\ j \notin S\}.
	\end{align*}

	By Equation~(\ref{eq:lexdiff}), only unions whose intersection with $\{i,j\}$ equals $\{i\}$ or $\{j\}$ contribute to the difference.
	Hence coalitions in $\mathfrak M^{+i,+j}(\Sigma)$ never affect the calculation, since every union containing one of them contains both $i$ and $j$.

	For the first sum in (\ref{eq:lexdiff}), only coalitions from $\mathfrak M^{+i,-j}(\Sigma)$ may contribute directly.
	As in Proposition~\ref{prp:opt1}, coalitions from $\mathfrak M^{-i,-j}(\Sigma)$ are only relevant after adjoining $i$, and only when they are not already covered by a coalition in $\mathfrak M^{+i,-j}(\Sigma)$.
	Hence it suffices to retain those sets $\{i\}\cup S$ with $S\in\mathfrak M^{-i,-j}(\Sigma)$ for which no $T\in\mathfrak M^{+i,-j}(\Sigma)$ satisfies $T\subseteq \{i\}\cup S$, and then take the inclusion-minimal members among them. Formally, define
	\[
		\mathcal F^{+i,-j}(\Sigma)
		=
		\MWC\bigl(
		\{
		\{i\}\cup S :
		S\in\mathfrak M^{-i,-j}(\Sigma),
		\nexists T\in\mathfrak M^{+i,-j}(\Sigma)\text{ with }T\subseteq \{i\}\cup S
		\}
		\bigr).
	\]
	Analogously, for the second sum it suffices to retain
	\[
		\mathcal F^{-i,+j}(\Sigma).
	\]

	Consequently, the running time is determined by the number of coalitions that can appear in the relevant antichains
	\[
		\mathfrak M^{+i,-j}(\Sigma)\cup \mathcal F^{+i,-j}(\Sigma)
	\]
	for the first sum, and symmetrically
	\[
		\mathfrak M^{-i,+j}(\Sigma)\cup \mathcal F^{-i,+j}(\Sigma)
	\]
	for the second sum.

	The worst case is attained by the extremal Sperner family
	\[
		\MWC(\Sigma)=\{S\subseteq N : \card{S}=\lfloor n/2\rfloor\}.
	\]
	In this case, every coalition in $\mathfrak M^{-i,-j}(\Sigma)$ has size $\lfloor n/2\rfloor$ and thus, after adjoining $i$ or $j$, strictly contains some coalition in $\mathfrak M^{+i,-j}(\Sigma)$ or $\mathfrak M^{-i,+j}(\Sigma)$, respectively. Therefore,
	\[
		\mathcal F^{+i,-j}(\Sigma)
		=
		\mathcal F^{-i,+j}(\Sigma)
		=
		\varnothing.
	\]

	Hence the maximal number of relevant coalitions in either sum equals
	\[
		\card{\mathfrak M^{+i,-j}(\Sigma)}
		=
		\card{\mathfrak M^{-i,+j}(\Sigma)}
		=
		\binom{n-2}{\lfloor n/2 \rfloor - 1}.
	\]
	since one chooses $\lfloor n/2 \rfloor - 1$ additional elements from $N\setminus\{i,j\}$.

	Each of the two sums in (\ref{eq:lexdiff}) therefore requires, in the worst case, enumeration of all nonempty subsets of such a family. Thus the total running time is
	\[
		O\!\left(2 \cdot 2^{\binom{n-2}{\lfloor n/2 \rfloor - 1}}\right).
	\]
\end{proof}

\end{document}